\documentclass[twocolumn,english,aps]{revtex4}
\usepackage[T1]{fontenc}
\usepackage[latin9]{inputenc}
\usepackage{xcolor}
\usepackage{amsthm}
\usepackage{amsmath}
\usepackage{amssymb}
\PassOptionsToPackage{normalem}{ulem}
\usepackage{ulem}

\makeatletter

\providecolor{lyxadded}{rgb}{0,0,1}
\providecolor{lyxdeleted}{rgb}{1,0,0}
\@ifundefined{textcolor}{}
{%
 \definecolor{BLACK}{gray}{0}
 \definecolor{WHITE}{gray}{1}
 \definecolor{RED}{rgb}{1,0,0}
 \definecolor{GREEN}{rgb}{0,1,0}
 \definecolor{BLUE}{rgb}{0,0,1}
 \definecolor{CYAN}{cmyk}{1,0,0,0}
 \definecolor{MAGENTA}{cmyk}{0,1,0,0}
 \definecolor{YELLOW}{cmyk}{0,0,1,0}
}
\theoremstyle{plain}
\newtheorem{thm}{\protect\theoremname}

\makeatother

\usepackage{babel}
\providecommand{\theoremname}{Theorem}

\begin{document}

\title{Quantum Control of Infinite Dimensional Many-Body Systems}

\author{Roger S. Bliss and Daniel Burgarth }

\affiliation{Institute of Mathematics, Physics and Computer Science, Aberystyth
University, Penglais Campus, SY23 2BZ Aberystwyth, United Kingdom}
\begin{abstract}
A major challenge to the control of infinite dimensional quantum systems
is the irreversibility which is often present in the system dynamics.
Here we consider systems with discrete-spectrum Hamiltonians operating
over a Schwartz space domain, and show that by utilizing the implications
of the Quantum Recurrence Theorem this irreversibility may be overcome,
in the case of individual states more generally, but also in certain
specified cases over larger subsets of the Hilbert space. We discuss
briefly the possibility of using these results in the control of infinite
dimensional coupled harmonic oscillators, and also draw attention
to some of the issues and open questions arising from this and related
work.
\end{abstract}
\maketitle

\section{Introduction}

One of the main goals of control theory in quantum mechanics is to
characterize the achievable dynamics of a system in relation to a
given set of control fields. While such \emph{reachable sets} are
well understood in the case of finite dimensional systems~\cite{DD08},
the same cannot be said about infinite dimensional systems. Apart
from the usual mathematical problems of infinite dimensional systems,
such as unbounded operators and the inequivalence of norms, the main
\emph{physical} obstacle is that infinite dimensional systems generally
show some form of \emph{irreversible} dynamics, which means they are
not controllable in general; and this causes further problems on a
mathematical level because the systems' reachable sets cannot be captured
\emph{algebraically} (that is, by looking at short-term dynamics only).
A typical way this obstacle manifests itself is in the \emph{drift
Hamiltonian, }which is the part of the system interaction which is
always present, independent of the applied controls, and which as
such can be the cause of irreversible dynamics.

From a physics perspective, the most relevant infinite dimensional
systems are those with position (and potentially spin) degrees of
freedom. In this case, it was shown in the seminal work by Braunstein
and Lloyd~\cite{SL99} that in the absence of a drift Hamiltonian,
any Hamiltonian can effectively be reached by switching between various
quadratic Hamiltonians and a single generic higher order one. Specific
simple periodic drift Hamiltonians were also considered. The work
in~\cite{SL99} was put in a rigorous mathematical framework in~\cite{RW06}.
Braunstein's result is very suitable for quantum optics, but does
not apply to systems with a constant many-body interactions. This
is for instance the case in systems of nano-mecanical oscillators~\cite{MA13}
and in solid state sytems described by Bose-Hubbard interactions~\cite{CB05}.
For these systems, specific examples of quantum control were developed
by~\cite{SM12} using the Baker-Campbell-Hausdorff formula and extensive
numerical analysis, but no general results regarding reachability
were given.

Recently, we developed a framework for quantum control of positive
definite quadratic Hamiltonians in the presence of drift Hamiltonians~\cite{MG12}.
Our analysis made use of the finite dimensional symplectic represenation
as well as the quantum recurrence theorem, and gave a general proof
of previous numerical observations~\cite{RW08}. Here, we obtain
a generalization of this work that only relies on Hamiltonians with
discrete eigenvalues, using the quantum recurrence theorem. Essentially,
such Hamiltonians naturally occur when the relevant system has only
bound states, so this case captures almost all relevant control applications
(see also~\cite{BS08} for a discussion of such Hamiltonians, relevant
to the cases discussed here).

Finally we apply our results to coupled oscillator systems and show
that the methods developed here allow us to substantially extend recent
results on indirect control~\cite{DB08,MG12} to the infinite dimensional
case.

\section{Setup}

For simplicity we will ignore spin degree of freedom and consider
wavefunctions in one spatial dimension only (the generalization is
straightforward). Although much previous work in infinite dimensional
control has been done treating closed loop systems~(see, for instance~\cite{JG09}
and references therein), we are focusing here on open loop controls.
We consider the Hilbert space $\mathcal{H}$ of square-integrable
wavefunctions (i.e., the space $L^{2}(\mathbb{R})$) and a finite
set of self-adjoint operators $\tilde{\mathcal{A}}=\left\{ \tilde{H}_{1},\tilde{H}_{2},\ldots,\tilde{H}_{K}\right\} $
which are assumed to be polynomials in position and momentum. The
set $\mathcal{\tilde{A}}$ describes the \emph{directly implementable
Hamiltonians}, i.e. the set of available Hamiltonians between which
the experimentalist may switch at will.\emph{ }Hence, a control sequence
consists of a list of times $(t_{1},t_{2},\ldots,t_{n}),\, t_{j}\ge0,$
and another of Hamiltonian choices $(k_{1},k_{2},\ldots,k_{n}),\,1\le k_{j}\le K,$
such that from time $0$ to $t_{1},$ the system evolves under $\tilde{H}_{k_{1}},$
from $t_{1}$ to $t_{2}$ under $\tilde{H}_{k_{2}}$ and so on. Note
that the set $\mathcal{\tilde{A}}$ may also describe systems with
a drift Hamiltonian $\tilde{H}_{0},$ e.g. by writing $\tilde{H}_{j}=\tilde{H}_{0}+H'_{j}$,
and that continuous control sequences can be described in a manner
similar to the piece-wise constant case discussed here. Assuming use
of natural units throughout, so that $\hbar=1$, we shall further
simplify notation hereafter by defining the set of skew-adjoint operators
$\mathcal{A}=$$\left\{ H_{1},H_{2},\ldots H_{K}\right\} $, where
$H_{j}=-i\tilde{H}_{j}$ for each $j=1,\ldots,K$.

As for domains, we follow ~\cite{RW06} and consider the Schwarz
space
\begin{equation}
\mathcal{H}_{s}=\left\{ \psi\in\mathcal{H}|\sup_{k,l\ge0}\left\Vert p^{k}q^{l}\psi\right\Vert <\infty\right\} ,
\end{equation}
where $\left\Vert \cdot\right\Vert $ is the Hilbert space norm and
$p$ and $q$ are the momentum and position operators, respectively.
This is a dense subspace of $\mathcal{H}$, on which $\mathcal{\tilde{A}}$
is defined. We will use the following theorems~\cite{RW06}: 
\begin{equation}
\lim_{n\rightarrow\infty}\left(e^{H_{k}\frac{t}{n}}e^{H_{\ell}\frac{t}{n}}\right)^{n}\psi=e^{(H_{k}+H_{\ell})t}\psi,\label{eq:Trotter}
\end{equation}
\begin{equation}
\lim_{n\rightarrow\infty}\left(e^{-H_{k}\frac{t}{n}}e^{-H_{\ell}\frac{t}{n}}e^{H_{k}\frac{t}{n}}e^{H_{\ell}\frac{t}{n}}\right)^{n^{2}}\psi=e^{\left[H_{k},H_{\ell}\right]t^{2}}\psi,\label{eq:Commutator}
\end{equation}
where $H_{k},H_{\ell}\in\mathcal{A}$ and $\psi\in\mathcal{H}_{S}.$
Note that to utilize Eq.~(\ref{eq:Commutator}) we need to be able
to employ $-H_{k}$ and $-H_{l}$ for $H_{k},H_{\ell}\in\mathcal{A}$.
How to do this when $-H_{k,}-H_{l}\notin\mathcal{A}$ is the central
problem addressed in this paper.

The reachable set $\mathcal{R}$ from an initial state $\psi_{0}\in\mathcal{H}_{S}$
is given by the states for which there exists a control such that
they become solutions to the Schrödinger equation with initial condition
$\psi_{0},$ e.g, 
\begin{equation}
\mathcal{R}\left(\psi\right)=\left\{ \prod_{j=1}^{n}e^{H_{j}t_{j}}\psi_{0}|H_{j}\in\mathcal{A},t_{j}\ge0,n\in\mathbb{N}\right\} .
\end{equation}
Similarly, we have the reachable set of unitaries, $\mathcal{U}$,
given by
\begin{equation}
\mathcal{U}=\left\{ \prod_{j=1}^{n}e^{H_{j}t_{j}}|H_{j}\in\mathcal{A},t_{j}\ge0,n\in\mathbb{N}\right\} .
\end{equation}
We define the dynamical Lie algebra $\mathfrak{l}_{s}\equiv\left\langle \mathcal{A}\right\rangle _{\left[\cdot,\cdot\right]}$
of the system as the real Lie algebra generated by $\mathcal{A}$,
i.e. the set of operators one obtains from $\mathcal{A}$ through
real linear combinations and commutators. Note that any element $\tilde{H}\in\mathfrak{l}_{s}/\left(-i\right)$
is a polynomial in position and momentum, and self-adjoint with domain
$D\left(\tilde{H}\right)\supset\mathcal{H}_{s}$.

\section{Controllability}

In this section we will provide some theorems that allow us to characterize
the controllability of infinite dimensional systems through \emph{reachable
sets. }First we use the quantum recurrence theorem to show that if
each element in $\mathcal{A}$ has a discrete spectrum, then irreversibility
can be overcome and with regard to the reachable set from a specific
starting state $\psi_{0}$ we have that $\bar{\mathcal{R}}\left(\psi_{0}\right)\supset e^{\mathfrak{l}_{s}}\psi_{0},$
where $\bar{\mathcal{R}}$ denotes the norm closure of $\mathcal{R}$~%
\footnote{Note that in infinite dimensions the closures of $\mathcal{R}$ and
$\mathcal{U}$ contain unphysical elements. This does not pose a problem
here as we are only interested in physical elements that can be approximated
arbitrarily closely with elements of $\mathcal{R}$ and $\mathcal{U}$.%
}. This result is analogous to the finite dimensional characterization~\cite{DD08}.
Because recurrence times are usually state dependent, such point-wise
characterizations can only be extended to larger sets using further
structure. To this end, we show that for any \emph{compact} set in
Schwarz space the closure of the set of reachable unitaries contains
$e^{\mathfrak{l}_{s}}$. Finally, motivated by physics, we consider
the case that each element in $\mathcal{A}$ has a spectrum of energy
eigenvalues which is bounded below and has no accumulation points.
We show in that case that any set of states with an overall bound
on the \emph{energy expectation values} with respect to $\mathcal{A}$
contains $e^{\mathfrak{l}_{s}}$ in the closure of its reachable set
of unitaries.

A key result employed here is the Quantum Recurrence Theorem~\cite{PB57,DW13}.
So long as our system is evolving under the influence of a constant
Hamiltonian that has a discrete spectrum of energy eigenvalues, then
the theorem tells us that if our system is in the state $\psi\left(T_{k}\right)$
at time $T_{k}$, then as the system evolves in time, there will for
any $\tau,\delta>0$ occur a time $T_{l}>T_{k}+\tau$ such that $\left\Vert \psi\left(T_{k}\right)-\psi\left(T_{l}\right)\right\Vert <\delta$
(where $\psi\left(T_{l}\right)$ is the state at time $T_{l}$). That
is, if one waits long enough from an initial time $T_{k}$ the system
will always return arbitrarily close to the state $\psi\left(T_{k}\right)$
in which it was at time $T_{k}$. Note as well that by the theory
of almost-periodic functions~\cite{HB47} the recurrence implied
by the theorem in fact represents an \emph{infinite} sequence of such
recurrences for any $\delta$.
\begin{thm}
\label{thm:If-each-element}If each element in $\mathcal{A}$ has
a discrete spectrum, then $e^{\mathfrak{l}_{s}}\psi_{0}\subset\bar{\mathcal{R}}\left(\psi_{0}\right)$.\end{thm}
\begin{proof}
For any $H_{k}\in\mathcal{A}$, we have $e^{H_{k}t_{k}}\psi_{0}\in\mathcal{R}\left(\psi_{0}\right)\subset\bar{\mathcal{R}}\left(\psi_{0}\right)$
for any $t_{k}\geq0$, by our definition of $\mathcal{R}\left(\psi_{0}\right)$
above. Consider now the case of $e^{-H_{k}t_{k}}\psi_{0}$ for $H_{k}\in\mathcal{A}$,
$t_{k}>0$. Now, because $H_{k}$ has a discrete spectrum we may invoke
the Quantum Recurrence Theorem to show that for any $\delta>0$ there
exists a $t_{k*}\geq0$ such that $\left\Vert \psi_{0}-e^{H_{k}t_{k}}e^{H_{k}t_{k^{*}}}\psi_{0}\right\Vert =\left\Vert \psi_{0}-e^{H_{k}\left(t_{k}+t_{k^{*}}\right)}\psi_{0}\right\Vert <\delta$,
where $t_{k}+t_{k^{*}}>0$. By the unitarity of the exponential operators,
this implies equivalently that $\left\Vert e^{-H_{k}t_{k}}\psi_{0}-e^{H_{k}t_{k^{*}}}\psi_{0}\right\Vert <\delta$.
As $e^{H_{k}t_{k^{*}}}\psi_{0}\in\mathcal{R}\left(\psi_{0}\right)$
for $t_{k^{*}}>0$, it follows that $e^{-H_{k}t_{k}}\psi_{0}\in\bar{\mathcal{R}}\left(\psi_{0}\right)$.
Now, for $H_{k,}H_{l}\in\mathcal{A}$, $t_{k,l}\in\mathbb{R}$, Eqs.~(\ref{eq:Trotter})
and~(\ref{eq:Commutator}) allow us to approximate $e^{\left(H_{k}+H_{l}\right)t_{k,l}}\psi_{0}$
and $e^{\left[H_{k},H_{l}\right]t_{k,l}}\psi_{0}$, respectively,
with an arbitrary degree of accuracy. It follows that $e^{\left(H_{k}+H_{l}\right)t_{k,l}}\psi_{0},e^{\left[H_{k},H_{l}\right]t_{k,l}}\psi_{0}\in\bar{\mathcal{R}}\left(\psi_{0}\right)$.
Combining the above we have that $e^{\mathfrak{l}_{s}}\psi_{0}\subset\bar{\mathcal{R}}\left(\psi_{0}\right)$. 
\end{proof}
The result of Theorem 1 is applicable to any $\psi_{0}\in\mathcal{H}_{S}$,
but in general only to such $\psi_{0}$ considered individually. In
the context of quantum computing this is especially problematic, firstly
because we would like in general to implement unitary operations not
just on individual states but across subsets of the Hilbert space,
and secondly because we may not in any case be able to possess perfect
information about the initial state. The ability to perform a unitary
operation reliably over a set of states close to the one of interest
thus provides a stability condition for the controllability of the
system. In the following two theorems we extend the result to larger
subsets of $\mathcal{H}_{S}$. In both cases the result depends on
the details of the proof of the Quantum Recurrence Theorem, outlined
below. The point is that while recurrence times usually depend on
the state (making the control sequence state-dependent in turn), under
the extra condition of compactness or finite energy expectation state-independent
recurrence times, and thus state-independent control sequences, may
be found.

As a solution to the Schrödinger equation at time $T_{k}$, with a
constant discrete-spectrum Hamiltonian, the state $\psi\left(T_{k}\right)$
may be expanded as the infinite sum $\psi\left(T_{k}\right)=\sum_{n=0}^{\infty}c_{n}e^{-iE_{n}T_{k}}\phi_{n}$,
where the $c_{n}$ are complex coefficients depending on $\psi\left(T_{k}\right)$
(with $\sum_{n=0}^{\infty}\vert c_{n}\vert^{2}=1$), and $\phi_{n}$
are the eigenstates corresponding to the Hamiltonian's discrete energy
eigenvalues $E_{n}$. Assuming the Hamiltonian remains constant, a
state $\psi\left(T_{l}\right)$ at time $T_{l}>T_{k}$ may be similarly
expanded as $\psi\left(T_{l}\right)=\sum_{n=0}^{\infty}c_{n}e^{-iE_{n}T_{l}}\phi_{n}$,
and then it follows that
\[
\left\Vert \psi\left(T_{k}\right)-\psi\left(T_{l}\right)\right\Vert ^{2}=2\sum_{n=0}^{\infty}\vert c_{n}\vert^{2}\left(1-\cos\left(E_{n}\tilde{T}\right)\right)
\]
 where $\tilde{T}=T_{l}-T_{k}$. Now, because $\sum_{n=0}^{\infty}\vert c_{n}\vert^{2}=1$,
it is possible for any $\delta>0$ to choose an $N\in\mathbb{N}$
such that
\begin{equation}
\sum_{n=N+1}^{\infty}\vert c_{n}\vert^{2}<\frac{\delta^{2}}{8}.\label{eq:QRThFact2}
\end{equation}
 Furthermore, by the theory of almost-periodic functions~\cite{HB47}
it will always be possible to find a value of $\tilde{T}$ (or indeed
an infinite sequence of such values) such that for any $\delta>0$
\begin{equation}
\sum_{n=0}^{N}\left(1-\cos\left(E_{n}\tilde{T}\right)\right)<\frac{\delta^{2}}{4}.\label{eq:QRThFact3}
\end{equation}
Combining these facts, with $N$ suitably chosen in Eq.~(\ref{eq:QRThFact2})
and $\tilde{T}$ suitably chosen in Eq.~(\ref{eq:QRThFact3}), gives
us
\begin{eqnarray*}
\lefteqn{\left\Vert \psi\left(T_{k}\right)-\psi\left(T_{l}\right)\right\Vert ^{2}}\\
 & = & 2\sum_{n=0}^{\infty}\vert c_{n}\vert^{2}\left(1-\cos\left(E_{n}\tilde{T}\right)\right)\\
 & \leq & 2\sum_{n=0}^{N}\left(1-\cos\left(E_{n}\tilde{T}\right)\right)+4\sum_{n=N+1}^{\infty}\vert c_{n}\vert^{2}\\
 & < & \frac{\delta^{2}}{2}+\frac{\delta^{2}}{2}=\delta^{2}
\end{eqnarray*}
from which it follows that $\left\Vert \psi\left(T_{k}\right)-\psi\left(T_{l}\right)\right\Vert <\delta$,
as desired. Note that in Eq.~(\ref{eq:QRThFact3}) the choice of
$\tilde{T}$ depends only on the $E_{n}$, which are in turn dependent
only on the Hamiltonian; i.e., the choice of $\tilde{T}$ is \emph{independent}
of the particular system state vectors under consideration.
\begin{thm}
For any compact set $X\subset\mathcal{H}_{S}$, $e^{\mathfrak{l}_{s}}\subset\bar{\mathcal{U}}$.\end{thm}
\begin{proof}
For $H_{k}\in\mathcal{A}$, $t_{k}\geq0$, we have $e^{H_{k}t_{k}}\in\mathcal{U}\subset\bar{\mathcal{U}}$.
Now consider the unitary $e^{-H_{k}t_{k}}$ for $H_{k}\in\mathcal{A},t_{k}>0$.
Let $\varepsilon>0$ be given, and let $\mathcal{B}=\left\{ B\left(\psi\right)|\psi\in X\right\} $,
where $B\left(\psi\right)=\left\{ \phi\in\mathcal{H_{S}}|\left\Vert \phi-\psi\right\Vert <\frac{\varepsilon}{3}\right\} $.
This is an open covering of $X$, so by compactness we may choose
a finite subcovering $\tilde{\mathcal{B}}=\left\{ B\left(\psi_{i}\right)|B\left(\psi_{i}\right)\in\mathcal{B},i=1,\ldots,q\right\} $.
Let $\delta=\frac{\varepsilon}{3}$ in Eqs.~(\ref{eq:QRThFact2})
and (\ref{eq:QRThFact3}) above. For each $\psi_{i}$ we may choose
an $N_{i}\in\mathbb{N}$ satisfying (\ref{eq:QRThFact2}); let $N=\max\left\{ N_{i}:i=1,\ldots,q\right\} $.
Having made this choice of $N$, we may choose a $\tilde{T}=t_{k}+t_{k^{*}}>0$
(with $t_{k^{*}}\geq0$) satisfying (\ref{eq:QRThFact3}), where the
choice of $\tilde{T}$ depends only on the energy eigenvalues of $\tilde{H}_{k}$.
With this choice of $N$ and $\tilde{T}$, it follows that $\left\Vert \psi_{i}-e^{H_{k}\left(t_{k}+t_{k^{*}}\right)}\psi_{i}\right\Vert <\frac{\varepsilon}{3}$,
or equivalently that $\left\Vert e^{-H_{k}t_{k}}\psi_{i}-e^{H_{k}t_{k^{*}}}\psi_{i}\right\Vert <\frac{\varepsilon}{3}$,
for all $\psi_{i},i=1,\ldots,q$. Now consider an arbitrary $\psi\in X$.
Noting that every such $\psi$ is contained in some $B\left(\psi_{i}\right)$,
we have that:

\begin{eqnarray*}
\lefteqn{\left\Vert e^{-H_{k}t_{k}}\psi-e^{H_{k}t_{k^{*}}}\psi\right\Vert }\\
 & \leq & \left\Vert e^{-H_{k}t_{k}}\psi-e^{-H_{k}t_{k}}\psi_{i}\right\Vert +\left\Vert e^{-H_{k}t_{k}}\psi_{i}-e^{H_{k}t_{k^{*}}}\psi_{i}\right\Vert \\
 &  & +\left\Vert e^{H_{k}t_{k^{*}}}\psi_{i}-e^{H_{k}t_{k^{*}}}\psi\right\Vert \\
 & = & \left\Vert \psi-\psi_{i}\right\Vert +\left\Vert e^{-H_{k}t_{k}}\psi_{i}-e^{H_{k}t_{k^{*}}}\psi_{i}\right\Vert +\left\Vert \psi_{i}-\psi\right\Vert \\
 & < & \frac{\varepsilon}{3}+\frac{\varepsilon}{3}+\frac{\varepsilon}{3}=\varepsilon.
\end{eqnarray*}
Thus we see that $e^{H_{k}t_{k}}\in\bar{\mathcal{U}}$. Now, noting
that neither Eq.~(\ref{eq:Trotter}) nor (\ref{eq:Commutator}) depend
on the system state vector $\psi$ under question, then for $H_{k,}H_{l}\in\mathcal{A}$
and $t_{k,l}\in\mathbb{R}$ we will have $e^{\left(H_{k}+H_{l}\right)t_{k,l}},e^{\left[H_{k,}H_{l}\right]t_{k,l}}\in\bar{\mathcal{U}}$,
for reasons analogous to those in the proof of Theorem 1. Combining
the above we have $e^{\mathfrak{l}_{s}}\subset\bar{\mathcal{U}}$.
\end{proof}
Finally we consider the case of states with bounded energy expectation
value.
\begin{thm}
If each element in $\tilde{\mathcal{A}}$ has a discrete spectrum
of energy eigenvalues which is bounded below and has no accumulation
points, then for any set of states $X\subset\mathcal{H}_{S}$ with
an overall bound on the energy expectation values with respect to
$\mathcal{\tilde{A}}$, $e^{\mathfrak{l}_{s}}\subset\bar{\mathcal{U}}$.\end{thm}
\begin{proof}
Again, $e^{H_{k}t_{k}}\in\mathcal{U}\subset\bar{\mathcal{U}}$ for
$H_{k}\in\mathcal{A}$, $t_{k}\geq0$, but we must consider the case
for $e^{-H_{k}t_{k}}$ with $H_{k}\in\mathcal{A}$, $t_{k}>0$. For
$H_{k}$ with energy eigenvalues bounded below, we may organize this
set of eigenvalues $\left\{ E_{0},E_{1},E_{2},\ldots,E_{n},\ldots\right\} $
in ascending order, i.e. such that $E_{0}\leq E_{1}\leq E_{2}\leq\ldots\leq E_{n}\leq\ldots$
for all $n\in\mathbb{N}$, where we assume with loss of generality
that $E_{0}\geq0$. For an infinite dimensional system this boundedness
of the eigenvalues below, combined with the lack of accumulation points,
also gives us that $E_{n}\rightarrow\infty$ as $n\rightarrow\infty$,
while by the bound on the energy expectation value we have$\left\langle \tilde{H}_{k}\right\rangle _{\psi}<M$
for all $\tilde{H}_{k}\in\tilde{\mathcal{A}}$, $\psi\in X$, and
for some finite $M\in\mathbb{R}$, $M>0$ (with $c_{n},E_{n}$ defined
as above). Let $\delta>0$ be given, and choose finite $N\in\mathbb{N}$
such that $E_{N+1}\geq\frac{8M}{\delta^{2}}$ (noting that the lack
of accumlation points makes such a choice possible). Since $E_{N+1}\leq E_{N+2}\leq\ldots$,
we have that 
\begin{eqnarray*}
\lefteqn{\sum_{n=N+1}^{\infty}\left|c_{n}\right|^{2}E_{N+1}=E_{N+1}\sum_{n=N+1}^{\infty}\left|c_{n}\right|^{2}}\\
 & \le & \sum_{n=N+1}^{\infty}\left|c_{n}\right|^{2}E_{n}\le\sum_{n=0}^{\infty}\left|c_{n}\right|^{2}E_{n}<M,
\end{eqnarray*}
from which it follows that $\sum_{n=N+1}^{\infty}\left|c_{n}\right|^{2}<\frac{M}{E_{N+1}}\leq\frac{\delta^{2}}{8}.$
Note that although this choice of $N$ satisfies Eq.~(\ref{eq:QRThFact2})
for every $\psi\in X$, the choice depends only on the energy eigenvalues
of $\tilde{H}_{k}$. Having thus chosen $N$, we may then make a choice
of $\tilde{T}=t_{k}+t_{k^{*}}$ (with $t_{k^{*}}\geq0$) satisfying
Eq.~(\ref{eq:QRThFact3}), where the choice of $\tilde{T}$ also
depends only on the energy eigenvalues of $\tilde{H}_{k}$. With this
$N$ and $\tilde{T}$, it follows that for all $\psi\in X$ we have
$\left\Vert \psi-e^{H_{k}\left(t_{k}+t_{k^{*}}\right)}\psi\right\Vert <\delta$,
or equivalently $\left\Vert e^{-H_{k}t_{k}}\psi-e^{H_{k}t_{k^{*}}}\psi\right\Vert <\delta$,
and therefore that $e^{-H_{k}t_{k}}\in\bar{\mathcal{U}}$. As in the
proof of Theorem 2, we will now have $e^{\left(H_{k}+H_{l}\right)t_{k,l}},e^{\left[H_{k,}H_{l}\right]t_{k,l}}\in\bar{\mathcal{U}}$
for $H_{k},H_{l}\in\mathcal{A}$ and $t_{k,l}\in\mathbb{R}$. Combining
the above we have $e^{\mathfrak{l}_{s}}\subset\bar{\mathcal{U}}$.
\end{proof}
We remark that the same argument to obtain universal recurrence times
was made independently in a different context by Wallace~\cite{DW13}.
For a consideration of recurrence in the presence of Hamiltonians
with discrete spectra, in the context of proving approximate full
controllability in certain finite and infinite dimensional systems,
see also \cite{TC08}.

\section{Example}

As an example we consider a generalization of~\cite{SL99} where
it was shown that cubic terms can generate polynomials of arbitrary
order. We will be brief in this presentation; details may be found
at \cite{DB08}. Let us look at a system of $N$ coupled harmonic
oscillators which are interacting through the Hamiltonian $\sum_{i,j}a_{ij}H_{ij}$,
$a_{ij}\ge0,$ where the interaction is given by coupled oscillators
with spring constant $\omega$ in the rotating wave approximation

\begin{equation}
H_{ij}=p_{i}^{2}+q_{i}^{2}+p_{j}^{2}+q_{j}^{2}+\omega\left(p_{i}-p_{j}\right)^{2}+\omega\left(q_{i}-q_{j}\right)^{2},
\end{equation}
and the coupling strengths $a_{ij}$ determine the geometry of the
system. If we denote the local Lie algebra of all skew-adjoint polynomials
on one oscillator $i$ by $\mathfrak{l}_{i}$, and the Lie algebra
of skew-adjoint polynomials on two oscillators $i$ and $j$ by $\mathfrak{l}_{ij}$,
then it follows easily from the canonical commutation relationships
that 
\begin{equation}
\left\langle \mathfrak{l}_{i},\left[\mathfrak{l}_{i},iH_{ij}\right]\right\rangle _{[\cdot,\cdot]}=\mathfrak{l}_{ij},
\end{equation}
i.e. $\mathfrak{l}_{ij}$ is the smallest Lie algebra that contains
$\mathfrak{l}_{i}$ and elements of the form $\left[\mathfrak{l}_{i},iH_{ij}\right]$.
This property is known as algebraic propagation~\cite{DB08} and
implies that an easy criterion can be used to show that the system
algebra $\mathfrak{l}_{s}$ is the complete one defined by the couplings
$a_{ij}$ and the controls. Using the fact that $H$ has a discrete
spectrum and our Theorem~\ref{thm:If-each-element}, we find that
a chain of oscillators controlled on one end is fully state controllable.

\section{Conclusions}

We have extended the well-known characterizations of reachable sets
in finite dimensions to a large and physically relevant class of infinite
dimensional systems. Our result raises many interesting issues. Firstly,
while in finite dimensions control sequences can in principle be computed
numerically, these algorithms become less and less efficient with
higher dimensions. While in some cases, going to infinite dimensions
will \emph{simplify }the numerics~\cite{RW08}, we expect that most
control pulses are \emph{uncomputable} with classical computers. Secondly,
it is a long-standing open problem~\cite{DD09} to obtain useful
bounds on \emph{how long }it takes to achieve a reachable operation.
While numerical examples in finite dimension often yield short times
in practice, we expect that the worst-case examples in finite dimensions
can scale as badly as recurrence times - reaching the lifetime of
the universe easily. It would be interesting to see how this changes
in infinite dimensions, and we conjecture that by introducing physical
constraints - such as finite energy expectations of the target states
- one can still find many practical controls (see also~\cite{SM12}).

\emph{Acknowledgements:} DB would like to thank M. Genoni, A. Serafini
and M. S. Kim for fruitful discussions about potential generalizations
of~\cite{MG12}. RSB would like to thank J. Gough, R. Gohm and D.
G. Evans for helpful discussions and suggestions on matters mathematical.

\end{document}